\documentclass[prd,twocolumn,floats,floatfix,showpacs,nofootinbib]{revtex4}
\usepackage{graphicx}
\usepackage{dcolumn}
\usepackage{bm}
\usepackage{etex}
\usepackage{graphics}
\usepackage{slashed}
\usepackage{amssymb}
\usepackage{natbib}
\usepackage{amsmath}
\usepackage{amsthm}
\usepackage{url}
\usepackage{color}
\usepackage{hyperref}

\usepackage{calrsfs}						
\usepackage[title,titletoc,toc]{appendix}
\newtheorem{theorem}{Theorem}
\newtheorem{proposition}{Proposition}
\usepackage{tikz-cd}

\newenvironment{example}[1][Example:]{\begin{trivlist}
\item[\hskip \labelsep {\bfseries #1}]}{\end{trivlist}}

\usepackage{fncylab}
\labelformat{equation}{(#1)}
\begin{document}
\title{Extended disformal approach in the scenario of Rainbow Gravity}

\author{Gabriel G.\ Carvalho}\email{Gabriel.Carvalho@icranet.org}
\author{Iarley P.\ Lobo}\email{Iarley.PereiraLobo@icranet.org}
\author{Eduardo Bittencourt}\email{Eduardo.Bittencourt@icranet.org}

\affiliation{CAPES Foundation, Ministry of Education of Brazil, Bras\'ilia, Brazil and\\
Sapienza Universit\`a di Roma - Dipartimento di Fisica\\
P.le Aldo Moro 5 - 00185 Rome - Italy}
\date{\today}

\begin{abstract}
We investigate all feasible mathematical representations of disformal transformations on a space-time metric according to the action of a linear operator upon the manifold's tangent and cotangent bundles. The geometric, algebraic and group structures of this operator and their interfaces are analyzed in detail. Then, we scrutinize a possible physical application, providing a new covariant formalism for a phenomenological approach to quantum gravity known as Rainbow Gravity.
\end{abstract}

\maketitle

\section{Introduction}
The study of new kinds of symmetries associated with equations of motion is crucial in modern physics, since it can elucidate hidden features there and help us find new nontrivial solutions to the involved equations \cite{dirac}. We quote Noether's theorem, gauge choices and the theory of the angular momentum operators, as some important examples from a lengthy list describing physical properties of a given system that can circumvent the cumbersomeness of solving equations of motion.

The same reasoning can be applied to disformal transformations. The increasing literature on this issue has revealed new physics beyond Bekenstein's initial proposal \cite{beken1, beken2}. Motivated by the results we have previously obtained in Ref.\ \cite{bitt15}, we show here that disformal transformations can be defined not only as purely geometric transformations, but they allow for distinct representations (algebraic, geometric and group) in general. In this way, we gather each mathematical aspect of disformal transformations of metric tensors in a unique object. By constructing an abstract operator acting upon the tangent spaces of a manifold, it is possible to see that it takes the form of a space-time geometry, a genuine algebraic tensor or a group element. This unified description of the disformal maps shall be illustrated by means of the diagrams in Sec.\ \ref{secmach}. The aforesaid operator singles out a preferred vector field (or a set of) to deform a previously defined tetrad frame. Therefore, depending on the choice of a such vector, the physical notions of time, energy and momentum can be altered. As we shall see later, new trends in quantum-gravity phenomenology, for instance Ref.\ \cite{magueijo04}, point to this direction.

Naturally, one can also interpret the new formalism developed here for disformal transformations as purely mathematical; nevertheless this is already sufficient to attract attention on its own. Notwithstanding, from the physical point of view, there has been a growing interest in disformal transformations due to their applications in several gravitational theories, for instance, Bekenstein's TeVeS formalism \cite{beken_mond}, which furnishes a covariant formulation for MOND \cite{milgrom} in the weak-field limit; bimetric theories of gravity \cite{clifton}, either scalar or scalar-tensor theories \cite{scalartheory,mota,dario,rua,sakstein1,matarrese,ip,sakstein2}, including Mimetic and Horndeski ones, disformal inflation \cite{nemanja} and analogue models to gravity \cite{nov_bit_gordon,nov_bit_drag}. There are also situations where disformal transformations are related to nonmetricity \cite{yuan} and particle physics, providing alternative explanations for the chiral symmetry breaking \cite{bitt_nov_faci}, the anomalous magnetic moment \cite{nov_bit}, as well as the disformal invariance of matter fields dynamics \cite{erico12,erico13,bitt15}.

This paper is organized as follows. In Sec.\ \ref{nowandthen} we review the standard geometric definitions for disformal transformations and, in Sec.\ \ref{newfacet}, we introduce the aforesaid disformal operator acting on vector fields on $M$. We then follow the same approach to recover the inverse of a disformal metric by introducing a operator acting on co-vectors. In other words, we establish a new procedure such that instead of mapping a previously defined {\it metric} onto a disformal metric, we map a previously defined {\it vector field} onto a disformal vector field. Indeed, a metric tensor is an additional structure we can endow a manifold with and the existence of vector fields only relies on the differential structure of the manifold, hence it is more fundamental. In section \ref{secmach} we show how these operators become algebraic tensors in terms of an arbitrary coordinate system and their relationship to the components of the disformal (co-)metric. With the help of the disformal group structure, already satisfied by disformal metrics, we then show that all the aforesaid operators can be reduced to a single one. In Sec.\ \ref{seccausal} we derive an algebraic criterium, in terms of the disformal parameters, to promptly find the causal relation between the background light cone and the disformal one. Also, we discuss how the introduction of a disformal operator generating a disformal metric gives rise to two consistent interpretations of causality. Finally, in Sec.\ \ref{applications} we scrutinize the case of Rainbow Gravity \cite{magueijo04}---a formalism to quantum-gravity phenomenology that takes into account a possible energy-dependence of the space-time metric probed by a particle with such energy---and show how such phenomenological approach has a natural explanation within the paradigm of disformal transformations.

\section{Disformal transformations: now and then}\label{nowandthen}
Using the definition presented in Ref.\ \cite{bitt15}, a single-vector disformal transformation of a given metric, which we shall simply call a disformal transformation, is an application that takes scalar functions $\alpha$ and $\beta$, a metric tensor $g$ and a globally defined time-like\footnote{It could be extended for light-like vectors or even tensorial fields as discussed in Refs.\ \cite{erico13,bitt15}, but these cases are out of the scope of this paper, for practical reasons.} vector field $V$ in a space-time and associates them to another (well-defined) metric tensor $\widehat{g}$ according to
\begin{equation}
\label{disformalmet1}
\widehat{g}(\ast,\cdot)=\alpha\, g(\ast,\cdot)+\frac{\beta}{g(V,V)}\,g(V,\ast)\otimes g(V,\cdot).
\end{equation}
The existence of a non-vanishing time-like vector field on $M$ is guaranteed when one considers time-orientable space-times, which is reasonable from the physical point of view. If $M$ is not time-orientable, there still exists a time-orientable twofold covering of $M$, where this formalism can lay upon (cf.\ brief discussion in Ref.\ \cite{penrose}). Taking that into account, besides mathematical convenience and physical reasonableness, we shall henceforth consider only time-orientable space-times.

The map which defines the disformal metric is well-defined if $\alpha>0$ and $\alpha+\beta>0$ throughout the manifold. This well-definiteness means that the constructed $\widehat{g}$ is a pseudo-Riemannian metric with the same signature as that of $g$. The scalars $\alpha$ and $\beta$ are not necessarily functions depending only on points of the manifold, but rather arbitrary functions that could also have a functional dependence on $V$ and its derivatives. We can write down explicitly the components of the disformal metric in a given coordinate system in its covariant and contra-variant versions as
\begin{eqnarray}
\widehat{g}_{\mu\nu}&=&\alpha g_{\mu\nu}+\frac{\beta}{V^2}V_{\mu}V_{\nu},\label{cov}\\[2ex]
\widehat{g}^{\mu\nu}&=&\frac{1}{\alpha}g^{\mu\nu}-\frac{\beta}{\alpha(\alpha+\beta)}\frac{V^{\mu}V^{\nu}}{V^2},\label{contrav}
\end{eqnarray}
where $V^2\equiv g_{\mu\nu}V^{\mu}V^{\nu}$ and it is straightforward to verify that $\widehat{g}_{\mu\nu}\widehat{g}^{\nu\sigma}=\delta_{\mu}^{\sigma}$.
\par
The main goal of the ensuing mathematical machinery is to introduce a new formalism to describe disformal metrics. This shall be done in terms of disformal operators acting on vectors and co-vectors in such a way that the disformal metric inherits their properties.

\subsection{New facet of a disformal transformation}\label{newfacet}
\par
Hereafter, let us fix a space-time $(M,g)$ and a non-vanishing time-like vector field $V \in \Gamma(TM)$\footnote{$\Gamma(TM)$ denotes the set of smooth sections of the tangent bundle, i.e. vector fields over $M$. This set is also denoted in the literature by $\mathfrak{X}(M)$ or $C^{\infty}(TM)$.}. Consider also any two scalars $\alpha$ and $\beta$ satisfying the conditions $\alpha>0$ and $\alpha+\beta>0$. Then, the map \ref{disformalmet1} can be equivalently described as an action on the tangent space in each point of the manifold by the following definition:
\begin{equation}
\widehat{g}\left(Y,Z\right)=g\left(\overrightarrow{D}(Y),\overrightarrow{D}(Z)\right),\label{deformednorm}
\end{equation}
where, for any $X\in\Gamma(TM)$, we have
\begin{equation}\label{def-oper}
\overrightarrow{D}(X)\doteq\sqrt{\alpha+\beta}X_{\parallel}+\sqrt{\alpha}X_{\perp},
\end{equation}
with
\begin{equation}
X_{\parallel}\doteq\frac{g(X,V)}{g(V,V)}V, \quad \mbox{and} \quad X_{\perp}\doteq X-\frac{g(X,V)}{g(V,V)}V,
\end{equation}
where $X_{\parallel}$ is the projection of $X$ onto $V$ and $X_{\perp}$ is the projection onto the orthogonal complement of $V$, such that $X=X_{\parallel}+X_{\perp}$. So, instead of working with a disformal transformation on the metric, one can define the map
\begin{eqnarray}
\overrightarrow{D} \ : \ \Gamma(TM)&\rightarrow & \Gamma(TM) \\
X&\mapsto & \widehat{X}=\overrightarrow{D}(X)\nonumber
\end{eqnarray}
to deform vectors and recover the disformal metric $\widehat{g}$. We shall denote by $\widehat{X}$ the disformal vector related to $X$ through the action of $\overrightarrow{D}$. Clearly, such map is linear, i.e. $\overrightarrow{D}(\gamma X+Y)=\gamma\overrightarrow{D}(X)+\overrightarrow{D}(Y)$ for any scalar function $\gamma$ and vectors $X$ and $Y$, and from this, $\overrightarrow{D}$ defines a mixed rank-2 tensor field on $M$ given by
\begin{eqnarray}
\overrightarrow{D}:\Gamma(T^{*}M)\otimes\Gamma(TM)&\rightarrow &\mathcal{F} (M)\\
(\theta,X)&\mapsto & \overrightarrow{D}(\theta,X)=\theta\left(\overrightarrow{D}(X)\right),\nonumber
\end{eqnarray}
where $\mathcal{F} (M)$ corresponds to the set of all smooth real-valued functions on $M$. Furthermore, $\overrightarrow{D}$ can be seen as a linear transformation on the tangent space $T_pM$ for each $p\in M$ and, provided $\alpha >0$ and $\alpha+\beta >0$, a linear isomorphism.
\par
It is simple to see that the operator $\overrightarrow{D}$ satisfying \ref{deformednorm} is not unique. In fact, all possible operators of the form $\overrightarrow{D}(X) = f_1 X_{\parallel} + f_2X_{\perp}$ satisfying \ref{deformednorm} are fourfold degenerated:
\begin{equation}
\label{poss_oper}
\overrightarrow{D}_{\{\pm,\pm\}}(X)=\pm\sqrt{\alpha+\beta}X_{\parallel}\pm\sqrt{\alpha}X_{\perp}.
\end{equation}
This degeneracy into the choice of $\overrightarrow{D}$ could lead to ambiguities in the definition of a disformal metric by a disformal operator according to Eq.\ \ref{deformednorm}. For reasons that shall become clear subsequently, we shall define the disformal operator as $\overrightarrow{D}(X)=\overrightarrow{D}_{\{+,+\}}(X)$, for every vector field $X$, and prove its uniqueness afterwards.

\subsection{The disformal co-metric}\label{seccometric}
In Sec.\ \ref{applications}, concerning the physical applications of our analyses, it will be important to use the disformal co-metric instead of the disformal metric; therefore, we briefly elaborate upon how one can analogously define a disformal operator acting on co-vectors to recover the information contained in Eq.\ \ref{contrav}.

For each vector $X$ in a manifold $M$ endowed with a metric tensor $g$, there exists its unique metric dual $\widetilde{X} \doteq g(X,\ast)$. Hence the dual is a linear map $\widetilde{X}:\Gamma(TM)\rightarrow \mathcal{F}(M)$, i.e. $\widetilde{X} \in \Gamma(T^{*}M)$. In this way, the co-metric $h$ is defined by a linear, symmetric and non-degenerate map:
\begin{eqnarray}
h:\Gamma(T^{\ast}M)\otimes \Gamma(T^{\ast}M)&\rightarrow& \mathcal{F}(M)\\
(\widetilde{X},\widetilde{Y})&\mapsto& h(\widetilde{X},\widetilde{Y})=g(X,Y).\nonumber
\end{eqnarray}
\par
In a given coordinate system $\{x^{\mu}\}$, the co-metric has components $g^{\mu\nu}=h(dx^{\mu},dx^{\nu})$, i.e., it is the contra-variant components of the metric. From Eq.\ \ref{contrav}, we can define the disformal co-metric intrinsically as
\begin{equation}
\widehat{h}(\ast,\cdot)=\frac{1}{\alpha}h(\ast,\cdot)-\frac{\beta}{\alpha(\alpha+\beta)} \frac{h(\widetilde{V},\ast)\otimes h(\widetilde{V},\cdot)}{h(\widetilde{V},\widetilde{V})}.\label{deformednorm2}
\end{equation}

Analogously to what we did before, we can define the disformal co-vector $\widehat{\omega}$ associated with $\omega$ by the application of a linear map $\widetilde{D}\ : \Gamma(T^{\ast}M)\rightarrow \Gamma(T^{\ast}M)$, from which we write the disformal co-metric as
\begin{gather}\label{deformedcometric}
\widehat{h}(\omega,\eta)=h\left(\widetilde{D}(\omega),\widetilde{D}(\eta)\right),
\end{gather}
with $\widetilde{D}$ given by
\begin{eqnarray}
\widetilde{D}(\omega)\doteq\frac{1}{\sqrt{\alpha+\beta}}\omega_{\parallel}+\frac{1}{\sqrt{\alpha}}\omega_{\perp}
\end{eqnarray}
and again we have the decomposition
\begin{equation}
\omega_{\parallel}\doteq\frac{h(\omega,\widetilde{V})}{V^2}\widetilde{V} \quad \mbox{and} \quad \omega_{\perp}\doteq\omega-\frac{h(\omega,\widetilde{V})}{V^2}\widetilde{V}.
\end{equation}
Similarly to the covariant disformal operator, we have four possible contra-variant disformal operators. We shall set $\widetilde{D}(\omega)=\widetilde{D}_{\{+,+\}}(\omega)$ to be the disformal operator for elements in $\Gamma(T^{*}M)$ and again prove that it is unique.

\section{Machinery and uniqueness of the disformal operators}\label{secmach}
In Ref.\ \cite{bitt15} one can find algebraic properties concerning disformal metrics as the eigenvalue problem for $\widehat{g}^{\mu}_{\ \nu}= \widehat{g}^{\mu\sigma}g_{\sigma\nu}$ and a group structure satisfied by them. Indeed, we now show that these metrics can be completely characterized by (disformal) operators, since they share similar properties. For practical purposes, it is useful to provide a coordinate representation for both $\overrightarrow{D} $ and $\widetilde{D}$. We then start with these coordinate expressions and use them to prove some propositions about the disformal operator, exploring their algebraic and geometric features.

\subsection{Coordinate expressions}\label{coordiexp}
To derive a coordinate expression for the disformal operators $\overrightarrow{D} $ and $\widetilde{D}$, let $\{ x^{\mu}\}$ be a coordinate system, $\{\partial_{\mu}\}$ the tangent vectors associated with the coordinate lines and $\{dx^{\mu}\}$ their duals. Thus,
\begin{eqnarray*}
\overrightarrow{D} (\partial_{\nu}) &=&\sqrt{\alpha}\,\partial_{\nu} + \frac{\sqrt{\alpha+\beta}-\sqrt{\alpha}}{V^2}V_{\nu}V.
\end{eqnarray*}
Applying $dx^{\mu}$ to this, we get the desired expression
\begin{equation}
\label{coord1}
\mathfrak{D}^{\mu}_{\ \nu}\doteq dx^{\mu}\left(\overrightarrow{D} (\partial_{\nu})\right) = \sqrt{\alpha}\, \delta^{\mu}_{\nu} + \frac{\sqrt{\alpha+\beta}-\sqrt{\alpha}}{V^2}V^{\mu}V_{\nu}.
\end{equation}
In coordinates, we thus have $\widehat{X}^{\mu} = \mathfrak{D}^{\mu}_{\ \nu}X^{\nu}$. Analogously, we get
\begin{equation}
\label{coord2}
\mathcal{D}^{\ \mu}_{\nu}\doteq \partial_{\nu}\left(\widetilde{D} (dx^{\mu})\right) = \frac{1}{\sqrt{\alpha}}\delta^{\mu}_{\nu} + \left( \frac{1}{\sqrt{\alpha + \beta}}-\frac{1}{\sqrt{\alpha }}\right)\frac{V^{\mu}V_{\nu}}{V^2},
\end{equation}
when defining $\widehat{\omega}_{\mu} = \mathcal{D}_{\mu}^{\ \nu}\omega_{\nu}$. It should be remarked that with our definitions vectors (co-vectors) transform upon the action of $\overrightarrow{D}$ ($\widetilde{D}$). Although it is possible to transform co-vectors (vectors) by means of $\mathfrak{D}^{\mu}_{\ \nu}$ ($\mathcal{D}^{\ \mu}_{\nu}$) this shall not be of our general interest here.

With the coordinate expressions for $\overrightarrow{D}$ and $\widetilde{D}$ at our disposal, we state:
\begin{proposition}
$\mathfrak{D}^{\mu}_{\ \nu}$ and $\mathcal{D}_{\nu}{}^{\mu}$ satisfy
\begin{gather}
\mathfrak{D}^{\mu}_{\ \sigma} \mathcal{D}^{\ \sigma}_{\nu}	=\delta^{\mu}_{\nu},
\end{gather}
hence acting as mutual inverses.
\end{proposition}
\begin{proof}
This is straightforward from Eqs.\ \ref{coord1} and \ref{coord2}.
\end{proof}
\par
Using Eqs.\ \ref{deformednorm}, \ref{deformedcometric} and the coordinate expressions for $\mathfrak{D}^{\mu}_{\ \nu}$ and $\mathcal{D}_{\nu}^{\ \mu} $, it is easy to show that
\begin{proposition}
The diagrams
\begin{equation*}
\begin{tikzcd}
\Gamma(TM) \arrow{r}{\overrightarrow{D}} \arrow[swap]{dr}{\widehat{g}} &
\Gamma(TM) \arrow{d}{g}\\
&\Gamma(T^{*}M)
\end{tikzcd}\hspace{10pt}
\begin{tikzcd}
\Gamma(T^{*}M)\arrow{r}{\widetilde{D}} \arrow[swap]{d}{\widehat{g}^{-1}} &
\Gamma(T^{*} M)\arrow{dl}{g^{-1}}\\
\Gamma(TM)
\end{tikzcd}
\end{equation*}
are not commutative.	
\end{proposition}
From the theory of differentiable manifolds it is known that for each point $p \in M$, the tangent $T_pM$ and cotangent $T_p^{*}M$ spaces of a differentiable manifold $M$ are naturally isomorphic linear spaces, although this isomorphism is basis-dependent. In the presence of a metric tensor on $M$, there is a canonical isomorphism between $T_pM$ and $T_p^{*}M$, namely, for a given vector $X\in T_pM$, there is a unique $\omega \in T_p^{*}M$ satisfying $g(X,\cdot) = \omega$. Since we are now dealing with a manifold endowed with two metric tensors, the above proposition becomes ambiguous when taking duals since the $g-$dual of a disformal vector is not the $\widehat{g}-$dual of the vector. This ambiguity is always present when the manifold under consideration has more than one metric tensor. Therefore, it is important to make clear which metric tensor is being used when raising and lowering indexes.

However, there is a commutative manner to deform vectors and co-vectors and take their duals. It is not difficult to show that the diagram
\begin{equation}\label{diagram}
\begin{tikzcd}
\Gamma(TM) \arrow{r}{\overrightarrow{D}} \arrow[swap]{d}{\widehat{g}} &
\Gamma(TM) \arrow{d}{g}\\
\Gamma(T^{*}M) \arrow[swap]{r}{\widetilde{D}} &\Gamma(T^{*}M)
\end{tikzcd}
\end{equation}
is {\it commutative}. In fact, we could have started with Eq.\  \ref{deformednorm} and the definition of $\overrightarrow{D}$ given by \ref{def-oper} and then define $\widetilde{D}$ to be the only operator able to make \ref{diagram} a commutative diagram. In doing so, one can recover the coordinate expression for $\widetilde{D}$ and the other properties associated with it. For completeness, we stress that if the direction of the arrows labeled by $g$ and $\widehat{g}$ is reversed, and replacing $g$ and $\widehat{g}$ by their inverses, the diagram is also commutative. Thus, one can start with the co-metric and the definition of $\widetilde{D}$ to define $\overrightarrow{D}$ and everything else in terms of it.

\subsection{Disformal group structure revisited}\label{group}
In this section we review one of the mathematical structures underlying disformal transformations (for more details the reader is addressed to Refs.\ \cite{erico13,bitt15}). Let $\overrightarrow{D}_i$ be disformal operators with disformal parameters $\alpha_i$ and $\beta_i$, for $i=1,2$, given by
\begin{gather}
\overrightarrow{D}_1(\cdot)=\sqrt{\alpha_1+\beta_1}(\cdot)_{\parallel}+\sqrt{\alpha_1}(\cdot)_{\perp},\\
\overrightarrow{D}_2(\cdot)=\sqrt{\alpha_2+\beta_2}(\cdot)_{\parallel}+\sqrt{\alpha_2}(\cdot)_{\perp}.	
\end{gather}
Since the action of $\overrightarrow{D}_i$ on a vector field is a also vector field, it is easy to verify that for any vector field $X$ holds
\begin{eqnarray}
\label{disformalaction}
\overrightarrow{D}_1\left(\overrightarrow{D}_2(X)\right)=\overrightarrow{D}_2\left(\overrightarrow{D}_1(X)\right)&=&\nonumber\\
 \sqrt{(\alpha_1+\beta_1)(\alpha_2+\beta_2)}X_{\parallel}+\sqrt{\alpha_1\alpha_2}X_{\perp}.
\end{eqnarray}
We can use this equation to define the composition of two disformal operators as
\begin{eqnarray}
\label{disformalgroup}
\left(\overrightarrow{D}_1\bullet \overrightarrow{D}_2\right)(X) =\left(\overrightarrow{D}_2\bullet\overrightarrow{D}_1\right)(X) &=&\nonumber\\
 \sqrt{(\alpha_1+\beta_1)(\alpha_2+\beta_2)}X_{\parallel}+\sqrt{\alpha_1\alpha_2}X_{\perp}.
\end{eqnarray}

One can also verify that the set of all disformal operators with the composition law given above is closed. It means that the composition of two disformal operators is itself another operator. The commutativity and associativity are easily checked, characterizing an Abelian group structure for that set, where the identity operator has parameters $\alpha=1$ and $\beta=0$ and the inverse of an operator with parameters $\alpha$ and $\beta$ has parameters $\alpha^{\prime} = \alpha^{-1}$ and $\beta^{\prime} = -\beta [\alpha (\alpha+\beta)]^{-1}$. A similar result is surely obtained for $\widetilde{D}$, since they share the same structure.

As mentioned before, there is a group structure associated with disformal metrics. Comparing the composition law \ref{disformalgroup} with the approach developed in Ref.\ \cite{bitt15}, it is neater and more elegant if we deal with operators instead of the group action. Finally, it should be noticed that particularly interesting examples of disformal sub-groups take place when all conformal coefficients are equal to $1$---which renders disformal metrics similar to those from the spin-2 field theory formulation, but with finite inverse metric---besides the cases in which the disformal coefficients are zero ($\beta's=0$), coinciding with the usual conformal group.

\subsection{Uniqueness of the disformal operator}\label{uniqueD}
It has been shown that disformal metrics are related to an Abelian group structure \cite{bitt15}. More precisely, there is an Abelian group acting on the space of metrics on $M$. Besides, in the previous section, we have seen that $\overrightarrow{D}$ and $\widetilde{D}$ also satisfy an Abelian group structure, and in Eq.\ \ref{deformednorm} we proposed that a disformal metric arises when we deform vectors and use the background metric. So, if we want to characterize the disformal metric in terms of a disformal operator, $\overrightarrow{D}$ must be well-defined and $\widehat{g}$ must inherit some of its properties. Recalling that there are four possible disformal operators satisfying \ref{deformednorm} and that we have claimed $\overrightarrow{D}_{\{+,+\}}$ is unique in a certain sense, we then state and prove the following
\begin{theorem}
$\overrightarrow{D}_{\{+,+\}}$ is the only disformal operator which satisfies \ref{deformednorm} and the disformal group structure. The same holds for the disformal co-metric and the operator $\widetilde{D}_{\{+,+\}}$.
\end{theorem}
\begin{proof}
Consider the set G of all admissible disformal operators given by Eq.\ \ref{poss_oper} with a composition law \ref{disformalaction}. By admissible we mean that whichever disformal parameters $\alpha$'s and $\beta's$ might satisfy $\alpha>0$ and $\alpha+ \beta>0$. For instance,
	\begin{eqnarray*}
	\left(\overrightarrow{D}^1_{\{-,+\}}\bullet \overrightarrow{D}^2_{\{+,-\}}\right)(X) = 	\left(\overrightarrow{D}^2_{\{+,-\}}\bullet\overrightarrow{D}^1_{\{-,+\}}\right)(X) =\nonumber\\
	 -\sqrt{(\alpha_1+\beta_1)(\alpha_2+\beta_2)}X_{\parallel}-\sqrt{\alpha_1\alpha_2}X_{\perp}=\overrightarrow{D}_{\{-,-\}}(X),
\end{eqnarray*}
for any $X\in\Gamma(TM)$, where the disformal parameters are $\alpha^{\prime}=\alpha_1\alpha_2$ and $\beta^{\prime} = \alpha_1\beta_2 +\beta_1\alpha_2 + \beta_1\beta_2$.
It is easy to ascertain that $(\mbox{G},\bullet)$ is an Abelian group isomorphic to the Klein four-group (also isomorphic to $\mathbb{Z}_2\oplus\mathbb{Z}_2$). As a result of this isomorphism, each element but the identity has order two and the only one whose square is itself is the identity; in our case these requirements are only fulfilled by$\overrightarrow{D}_{\{+,+\}}$. Therefore, the only closed composition of disformal operators holds when they are both of the form $\overrightarrow{D}_{\{+,+\}}$. The proof that it satisfies a group structure was given in the previous section. This is entirely analogous to the $\widetilde{D}$ operator.
\end{proof}

A geometrical argument to rule out the other candidates for the disformal operator (the ones with at least one negative sign) is: because of the negative sign, they must include a reflection in the direction perpendicular to $V$ and/or a change in direction along with $V$. Therefore, whenever they are applied an even number of times, a positive sign must appear, implying that the operation is not closed. Thus, this ensures that there is a unique and well-defined disformal operator characterizing the disformal metric.

\subsection{The disformal operator as the square root of the disformal metric}\label{squareroot}
We have shown that given a local coordinate system $\{x^{\mu}\}$ the coordinate expression for the disformal operator $\overrightarrow{D}$ takes the form \ref{coord1}. Lowering the $\mu$ index with the background metric $g$, we thus get
\begin{equation*}
\mathfrak{D}_{\mu\nu}\doteq g_{\mu\sigma} \mathfrak{D}^{\sigma}_{\ \nu}=\alpha^{\prime}g_{\mu\nu} +\frac{\beta^{\prime}}{V^2}V_{\mu}V_{\nu},
\end{equation*}
where $\alpha^{\prime}=\sqrt{\alpha} >0$ and $\alpha^{\prime} + \beta^{\prime} = \sqrt{\alpha+\beta}>0$ and, therefore, $\mathfrak{D}_{\mu\nu}$ can be seen as a disformal metric tensor on $M$.

Using the composition law between disformal metrics (given in Ref.\ \cite{bitt15})
\begin{eqnarray*}
\Big(\lfloor \alpha_1,\beta_1\rfloor \odot	\lfloor \alpha_2,\beta_2\rfloor \Big)g& =& (\alpha_1\alpha_2) g_{\mu\nu} \\
&+& \frac{\beta_1\alpha_2+\beta_1\beta_2+\alpha_1\beta_2}{V^2}V_{\mu}V_{\nu},
\end{eqnarray*}
we obtain that the square of $\mathfrak{D}$ is precisely $\widehat{g}$:
\begin{equation}
\Big(\lfloor \alpha^{\prime},\beta^{\prime}\rfloor \odot	\lfloor \alpha^{\prime},\beta^{\prime}\rfloor \Big)g= \widehat{g}_{\mu\nu}.
\end{equation}
Another way to see this is by looking at the eigenvalues of $\widehat{g}^{\mu}_{\ \nu} = \widehat{g}_{\sigma\nu}g^{\sigma\mu}$ and $\mathfrak{D}^{\mu}_{\ \nu}$: from the definition \ref{def-oper}, the eigenvalue problem associated with the operator $\overrightarrow{D}$ is trivially solvable. We see that $V$ is an eigenvector related to the eigenvalue $\lambda_{V}=\sqrt{\alpha + \beta}$, while the other eigenvalues are degenerated and equal to $\sqrt{\alpha}$, with linearly independent eigenvectors lying on the orthogonal complement of $V$. The same analysis can be carried out to $\widetilde{D}$. It means that the eigenvalues of $\widehat{g}^{\mu}_{\ \nu}$ (see Ref.\ \cite{bitt15}) are exactly the eigenvalues of $\mathfrak{D}^{\mu}_{\ \nu}$ squared, and hence $\mathfrak{D}^2 = \widehat{g}$ as operators.

\section{Remarks on the causal structure}\label{seccausal}
We now study the relationship between light cones of the background geometry and the disformal one, and how Eq.\ \ref{deformednorm} gives rise to two interesting, and equivalent, interpretations of causality in the context of disformally related metrics and operators. For the first task, consider a disformal metric as in Eq.\ \ref{disformalmet1}. Let us fix a point $p\in M$ and at that point consider an orthonormal basis $\{e_{A}\}$ with respect to the background metric $g$, where $e_{0} = V/\sqrt{g(V,V)}$. Thus, for any $X\in T_p M$, we can write $X = X^{A} e_{A}$ and obtain
\begin{eqnarray*}
	\widehat{g}(X,X) = (\alpha +\beta) \left(X^{0} \right)^2 - \alpha\,\delta_{IJ} X^IX^J,
\end{eqnarray*}
where $\delta_{IJ}$ is the Kronecker delta, for $I,J=1,2,3$.
\par
Since we want to compare light cones of both metrics, let us assume that $X$ is a null-like vector with respect to $\widehat{g}$, that is $\left(X^0\right)^2 - \,\delta_{IJ} X^IX^J = -\beta (X^0)^2/\alpha$. Therefore, at $p \in M$, we have the following conditions:
\begin{enumerate}
\item If $\beta =0$, then $X$ is also a null-like vector with respect to $g$ and the disformal light cone is the same as the background one;
\item If $\beta < 0$, then $X$ is a time-like vector with respect to $g$ and the disformal light cone lies inside the background one;
\item If $\beta>0$, then $X$ is a space-like vector with respect to $g$ and the background light cone lies inside the disformal one.
\end{enumerate}

For the second task, note that the existence of the tensor field $\overrightarrow{D}$, such that a disformal metric in Eq.\ \ref{disformalmet1} can be written as Eq.\ \ref{deformednorm} for any fields $X,Y$, allows us to interpret the left hand side of \ref{deformednorm} as a new metric tensor for $M$. Thus, the light cones of the metric $\widehat{g}$ are, in general, different from those of $g$ and hence have a different causal structure on $M$.

For practical reasons, we provide now a simple and merely illustrative example, without any physical meaning a priori. Let us fix the background metric to be the Minkowski one [$\eta=\mbox{diag} (1, -1, -1, -1)$] and, at a fixed point $p$, we have $V(p) = (2,1,0,0)$, $X(p)=5\,(-1,1,1,0)$, $\alpha(p) = 2 $ and $\beta(p)=3$. It is easily verified that $X$ is space-like with respect to $g$ and time-like with respect to $\widehat{g}$ at $p$, which means that a vector field could be space-like in the background metric and time-like in the disformal one. This situation is depicted on top of Fig.\ \ref{fig1}. Conversely, the right hand side of Eq.\ \ref{deformednorm} indicates that we can consider just the background metric $g$, but applied to deformed vectors. Therefore, although light cones are preserved, causal relations change because the vectors have done so. At the bottom of Fig.\ \ref{fig1}, the vector is originally space-like and its disformal counterpart is time-like. Surely, the important feature is that the causal relation must agree whether you apply the disformal metric to vectors or the background metric to the deformed vectors.

\begin{figure}
\begin{center}
\includegraphics[width=70mm]{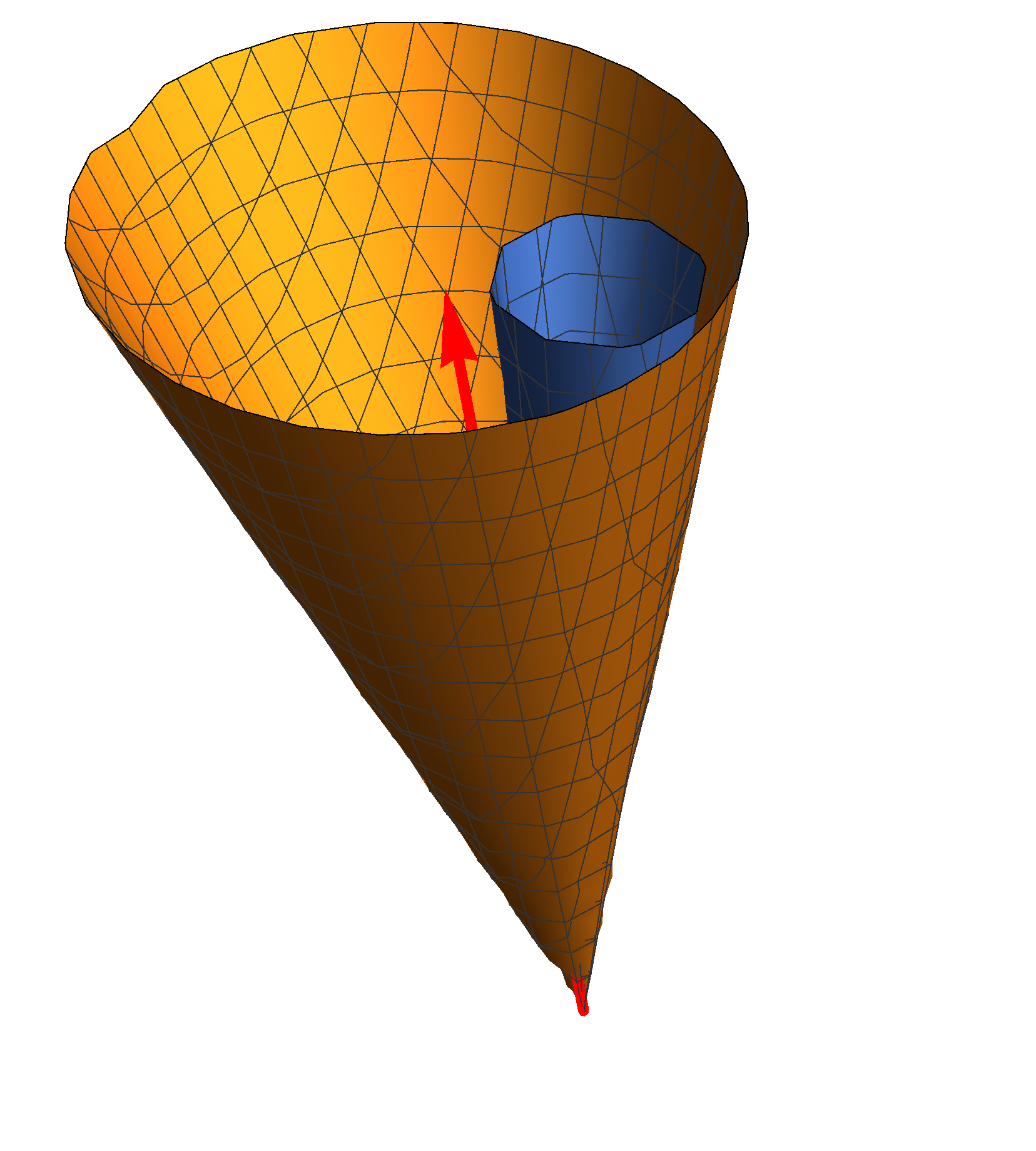}
\includegraphics[width=70mm]{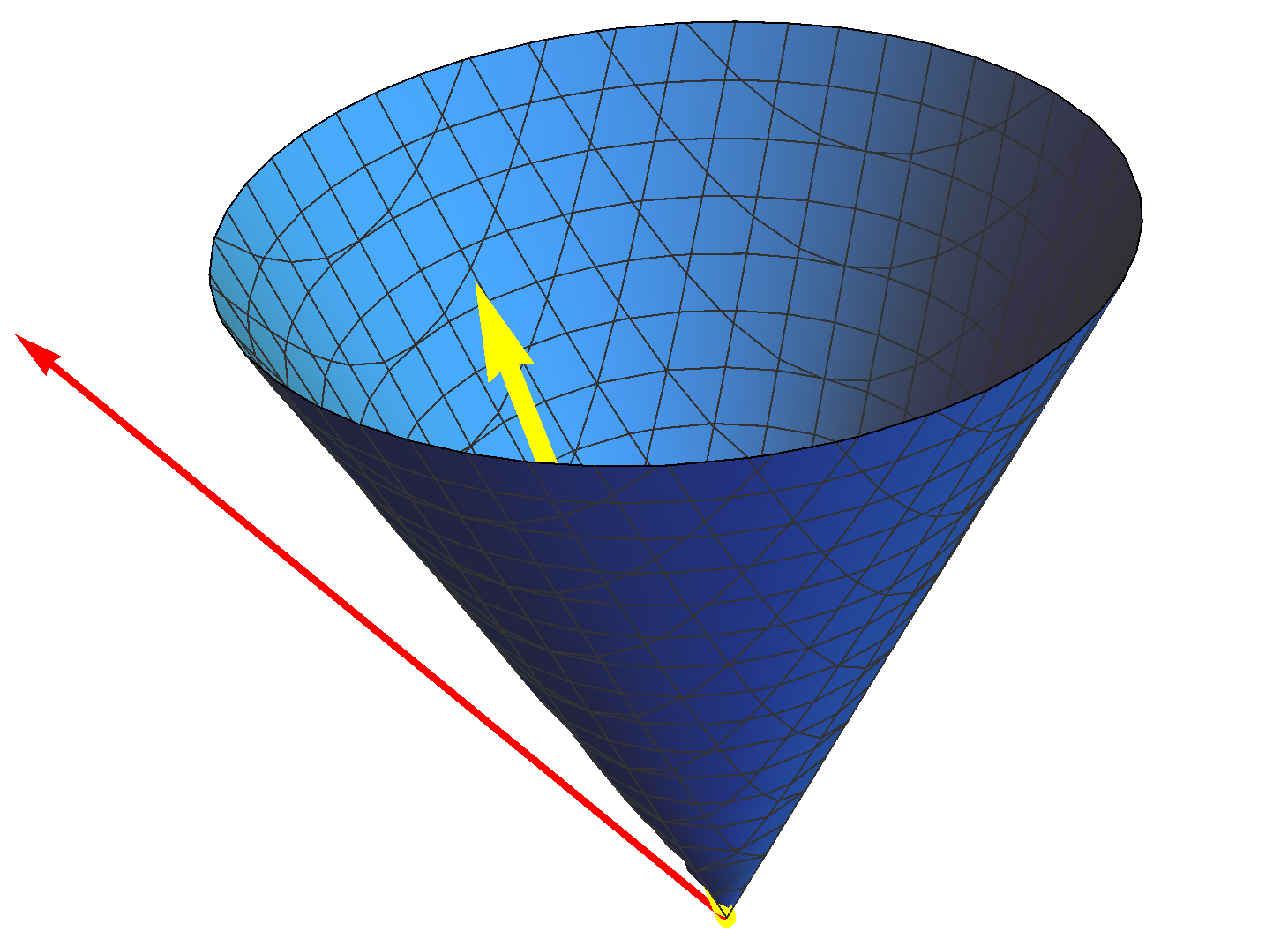}
\end{center}
\caption{(color online). On top, the vector is kept and the light cone is changed. In this case, $X$ is space-like in $g$ and time-like in $\widehat{g}$. At the bottom, the light cone is kept and the vector is deformed. The original vector (red) is space-like and its deformed counterpart (yellow) is time-like.}
\label{fig1}
\end{figure}
We thus conclude the mathematical aspects of disformal transformations we intended to develop here. Now, we shall see how to apply this simple and elegant formalism to the realm of quantum-gravity phenomenology.

\section{Application to Rainbow Gravity} \label{applications}
\par
As we have stated previously, disformal transformations can naturally be seen as pure mathematics, attractive on its own, and with various physical applications. In particular, we have seen here different representations for performing disformal transformations in the context of differential geometry, essentially by means of metric tensors.

Nevertheless, we are interested in the applications of the aforesaid formalism in what concerns phenomenological approaches to quantum-gravity, particularly, Rainbow Gravity \cite{magueijo04}. We believe that disformal transformations as presented above provide a unified language for deforming a background space-time metric in this scenario and can shed light on some fundamental problems there, like covariance and causality. With this in mind, we dedicate the forthcoming sections to discussing these issues.

\subsection{Rainbow Gravity and disformal metrics}\label{GR}
The formulation of Rainbow Gravity is a phenomenological modification of General Relativity that incorporates some properties of the Doubly Special Relativity (DSR) program \cite{magueijo04}. DSR models deform the kinematics of Special Relativity, modifying also the energy-momentum conservation laws and the Lorentz symmetry group, by admitting an invariant energy scale associated with quantum-gravitational effects: the Planck scale. The motivation behind this stems from the path used to go from Galilean Relativity to the Special Relativity, modifying the kinematic equations of the former in order to appear an invariant velocity scale. Following the same lines, it is possible to deform the latter by taking into account an invariant energy scale, which is generally believed to correspond to quantum-space-time effects, and thus derive the DSR without violation of the relativity principle. For pioneering works, see Refs.\ \cite{amelino01, smolin01, smolin02} and, for a broad review, see ref.\ \cite{amelino02}.

Following the prescription presented in \cite{smolin01} to perform such modifications, one can deform the momentum space of a particle with momentum $\pi=(p_0,p_i)$ using a function $U$ that depends on the ratio between the particle energy $p_0$ and the Planck energy $\kappa$, as follows\footnote{We consider geometric units: $c=\hbar=1$.}
\begin{equation}
\label{u-map}
U (p_0,p_i) = (f_1(p_0/\kappa)p_0, f_2(p_0/\kappa)p_i),
\end{equation}
leading to the modified dispersion relation (MDR):
\begin{equation}
\label{p-norm-mdr}
\|\pi\|^2=\eta^{AB}\left[U(\pi)\right]_A\left[U(\pi)\right]_B
=(f_1)^2p_0^2-(f_2)^2|\vec{p}|^2.
\end{equation}

In order to guarantee the invariance of this MDR, Lorentz symmetry transformations also need to be deformed. Although this deformation was initially intended to take place in the Minkowski space, the idea of Rainbow Gravity is that such MDR can be described by energy-dependent tetrad fields, which in turn produce an energy-dependent (rainbow) metric of the form
\begin{equation}
\label{line}
ds^2=\frac{(dx^0)^2}{f_1^2}-\frac{1}{f_2^2}\,\delta_{ij}\,dx^idx^j,
\end{equation}
where $\delta_{ij}$ is the Kronecker delta, for $i,j=1,2,3$. This means that space-time is deformed in the inverse way of the momentum space (for details, see Ref.\ \cite{magueijo04}). The $U-$transformation defined in \ref{u-map} resembles the ones we have considered throughout this paper, for a suitable choice of the disformal operator $\widetilde{D}$.

In fact, considering a time-like 1-form field $\widetilde{V}$ as defining a preferred direction in space-time, this leads to the definition of energy as the projection of the four-momentum $\pi$ onto the direction of the corresponding normalized 1-form vector $\nu \doteq \widetilde{V} / \sqrt{h(\widetilde{V}, \widetilde{V})}$, that is $p_0\doteq h(\pi,\nu)$. Therefore, the co-vector responsible for the disformal transformation introduces a natural time-like direction to the reference frame. Thus, using the orthonormal basis $\{\nu,\theta^I\}$, an immediate conclusion one can get from this analysis is that the disformal momentum assumes the form
\begin{equation}
\widehat{\pi} = \widetilde{D}(\pi) = \frac{1}{\sqrt{\alpha+\beta}}\,p_0\,\nu + \frac{1}{\sqrt{\alpha}}\,p_I\,\theta^I,
\end{equation}
where $\alpha$ and $\beta$ are now scalar-functions depending on $\pi$ and on $\nu$, and are linked to the rainbow functions $f_1$ and $f_2$ through Eq.\ \ref{u-map}:
\begin{equation}
\label{functions}
\alpha=(f_2)^{-2} \quad \mbox{and} \quad \beta=(f_1)^{-2}-(f_2)^{-2}.
\end{equation}
Furthermore, from the definition of the particle mass as the norm of $\pi$, a MDR naturally appears from this map in complete analogy with Eq.\ \ref{p-norm-mdr}:
\begin{equation}
\widehat{m}_{\pi}^2\doteq\widehat{h}(\pi,\pi)=\left(\frac{1}{\alpha+\beta}\right)p_0^2-\frac{1}{\alpha}\,\delta^{IJ}\,p_I\,p_J.
\end{equation}
Finally, the equivalence between the formalism we developed here and Rainbow Gravity is fulfilled by deriving the induced space-time metric [see Eq.\ \ref{line}]
\begin{equation}
d\widehat{s}^2=\widehat{g}_{\mu\nu}dx^{\mu}dx^{\nu}=(\alpha+\beta)(dx^0)^2-\alpha\,\delta_{ij}\,dx^idx^j.
\end{equation}

Thus, we could identify unequivocally the energy that appears in \ref{u-map} as $p_0=h(\pi,\nu)$ well as the respective time-like direction that defines the deformation. We stress that although this formalism seems to impose a preferred inertial frame in space-time, which would break the local canonical Lorentz symmetry, in the light of a DSR formulation this is not at all the case, since a {\it deformed version of the Lorentz transformation is the one that preserves the local relativity principle}: this is the main conceptual achievement of DSR. Were not the existence of deformations taken into account, then this would lead to a formalism with Lorentz invariance violation and, consequently, a preferred reference frame. It should also be noticed that the formalism we developed here is intrinsically geometric and then it is fully covariant under coordinate transformations.


\subsection{Some examples}\label{examples}
We now make use of the literature in Rainbow Gravity (cf.\ Refs.\ \cite{smolin01, grasi}) to illustrate with some examples how this relation works in practice.

\begin{example}[The case with $f_1=f_2$:] it will not alter the light cone. If $f_1$ is equal to $f_2$, then $\beta=0$ and the disformal transformation reduces to a conformal one. A well-known example on this was first proposed in Ref.\ \cite{smolin01}:
\begin{equation}
f_1(E/\kappa)=f_2(E/\kappa)=\frac{1}{1-E/\kappa}
\end{equation}
implying that
\begin{equation}
\alpha =\left[1-\frac{h(\pi,\nu)}{\kappa}\right]^{2}.
\end{equation}
This choice of deformation yields a maximum energy for a one-particle system, given by $\kappa$ and the causal structure is maintained invariant of course.
\end{example}
\begin{example}[The case with $f_2=1$:] this second example (cf.\ details in Ref.\ \cite{grasi}) has an invariant spatial contribution for the dispersion relation. Let
\begin{equation}
f_1(E/\kappa)=\frac{e^{E/\kappa}-1}{E/\kappa} \quad \mbox{and} \quad f_2(E/\kappa)\equiv 1.
\end{equation}
In terms of the disformal functions, we get
\begin{equation}
\alpha =1,\quad \mbox{and} \quad
\beta =\left(\frac{h(\pi,\nu)/\kappa}{e^{h(\pi,\nu)/\kappa}-1}\right)^2-1.
\end{equation}
\par
For this dispersion relation, one can calculate the speed of light as $(dE/dp)|_{m=0}\approx 1-E/\kappa$. Therefore, ultra-violet photons propagate with speed smaller than infra-red ones, within this model. Note that this is completely consistent with the causal structure analyzed in Sec.\ \ref{seccausal}. Since $-1<\beta< 0$, we have that $\alpha+\beta<\alpha$ and, therefore, the disformal light cone lies inside the undeformed one.
\end{example}

\begin{example}[The case with $f_1=1$:]	this third example is the opposite of the previous one (see Ref.\ \cite{dimensional} and references therein), in the sense that the time contribution is now kept invariant. Consider
\begin{equation}
f_1=1 \quad, \mbox{and} \quad  f_2=\left[1+\left(\frac{|\vec{p}|}{\kappa}\right)^4\right]^{\frac{1}{2}}.
\end{equation}
In terms of the metric coefficients, we obtain
\begin{subequations}
\begin{eqnarray}
\alpha &=&\frac{\kappa^4}{k^4+[h^2(\pi,\nu)-h(\pi,\pi)]^2}\\[2ex]
\beta &=&\frac{[h^2(\pi,\nu)-h(\pi,\pi)]^2}{\kappa^4+[h^2(\pi,\nu)-h(\pi,\pi)]^2}.
\end{eqnarray}
\end{subequations}
\par
For this dispersion relation the deformed speed of light is $(dE/dp)|_{m=0}\approx 1+5(p/\kappa)^4/2$, which means that high-energy photons propagate with speed larger than low-energy ones. Again, this is compatible with the causal structure, once $0< \beta <1$ and, consequently, $\alpha+\beta>\alpha$. Therefore, the disformal light cone lies outside the undeformed one.
\end{example}

\section{Concluding remarks}\label{conclusion}
We have shown that disformal metrics can be written in terms of a linear isomorphism acting on the tangent space and that they actually inherit the properties of what we called the disformal operator. From the reasons presented in the text, this operator can be seen as a more fundamental quantity than the disformal metric, providing a mathematical framework for disformal transformations. We then analyze this new facet of disformal transformations in the light of the causal structure, where it gives rise to an alternative interpretation of the modified causal cones in purely algebraic terms.

Finally, as a direct application of the formalism developed previously, we verified that the most relevant models in Rainbow Gravity are perfectly described in terms of disformal transformations. In this vein, it was possible to obtain the missing covariant approach for such phenomenological theory, with a well behaved causal structure and a clear mathematical interpretation of the physical quantities involved.

\section*{Acknowledgments}
We would like to express our gratitude to Jonas P. Pereira and Grasiele B. dos Santos for reading a previous version of this manuscript and suggesting several valuable improvements. The authors are supported by the CAPES-ICRANet program (BEX  15114/13-9, 14632/13-6, 13956/13-2).

\end{document}